\documentclass[12pt]{article}

\usepackage{float}
\usepackage{latexsym}
\usepackage{amsmath}
\usepackage{amssymb}
\usepackage{amsthm}
\usepackage{epsfig}
\usepackage{ulem}
\usepackage{graphicx}
\usepackage{layout}
\usepackage{manfnt}
\usepackage{algorithm}
\usepackage{authblk}
\usepackage{algorithmic}
\usepackage{enumerate}
\usepackage{esvect}
\usepackage{graphicx}
\usepackage{fullpage}
\usepackage{subfigure}
\usepackage{graphics}

\usepackage{cite}
\usepackage[colorlinks=true,linkcolor=red, citecolor=blue]{hyperref}

\newtheorem{theorem}{Theorem}
\newtheorem{lemma}{Lemma}
\author{
  Rehan Akber\thanks{Email: rehan.akber@lums.edu.pk},
  \quad
  Adnan Khan\thanks{Email: adnan.khan@lums.edu.pk}
}

\date{
  Department of Mathematics, Lahore University of Management Sciences (LUMS), Lahore, Pakistan
}

\begin{document}
\title{Mathematical Modeling of Biofilm Eradication Using Optimal Control}

\maketitle
\section*{Abstract}
We propose and analyze a model for antibiotic resistance transfer in a bacterial biofilm and examine antibiotic dosing strategies that are effective in bacterial elimination. In particular, we consider a 1-D model of a biofilm with susceptible, persistor and resistant bacteria. Resistance can be transferred to the susceptible bacteria via horizontal gene transfer (HGT), specifically via conjugation. We analyze some basic properties of the model, determine the conditions for existence of disinfection and coexistence states, including boundary equilibria and their stability. Numerical simulations are performed to explore different modeling scenarios and support our theoretical findings. Different antibiotic dosing strategies are then studied, starting with a continuous dosing; here we note that high doses of antibiotic are needed for bacterial elimination. We then consider periodic dosing, and again observe that insufficient levels of antibiotic per dose may lead to treatment failure. Finally, using an extended version of Pontryagin's maximum principle we determine efficient antibiotic dosing protocols, which ensure bacterial elimination while keeping the total dosing low; we note that this involves a  tapered dosing which reinforces results presented in other clinical and modeling studies. We study the optimal dosing for different parameter values and note that the optimal dosing schedule is qualitatively robust. 

\section{Introduction}

Extensive and improper use of antibiotics has led to the spread of antibiotic resistance which is now considered to be a major global public health challenge \cite{ADN1, ADN2}. A recent study estimated  1·27 million deaths were attributable to bacterial antimicrobial resistance (AMR) in 2019 with projections of 1.91 million deaths due to AMR in 2050 \cite{ADN3, P6_1,P6_2,P6_3}. 

Antibiotic resistance occurs when bacteria are able to proliferate in the presence of antibiotic agents at specific and previously effective concentration levels. The problem has been exacerbated by the fact that the development of new antibiotics has been outpaced by the acquisition of resistance to existing drugs by bacteria in recent years.

Antibiotic resistance in bacteria occurs either by spontaneous mutations in the DNA (De-Novo resistance) or through the transfer of genetic material from drug-resistant bacteria via horizontal gene transfer (HGT). The overuse and misuse of antibiotics provide selection pressure favoring the resistant strains causing these to proliferate. Unlike vertical gene transfer, in which genes are transferred from parents to offspring, HGT allows bacteria or other organisms to receive new genes from another organism directly.  Well known mechanisms of HGT include conjugation which involves the transfer of DNA via a plasmid during cell to cell contact\cite{P6_6,P6_7}, transformation where  foreign genetic material is taken up and expressed by a bacterium\cite{P6_8,P6_9}, and transduction in which  DNA is moved from one bacterium to another by a virus (phage) \cite{P6_10}. 

Bacterial biofilms are colonies of bacteria embedded within a self produced extracellular matrix and it has been observed that biofilms contribute significantly to the spread of antimicrobial resistance.  Clinical infections are increasingly linked to biofilms that can tolerate antibiotics at concentrations 100–10,000 times higher than those required to kill planktonic cells \cite{P5_2,P5_3}. The existence of biofilms was first reported by Anton Van Leeuwenhoek in the 17th century, based on his microscopic observations of dental plaque, modern research on this started in 1970s \cite{P7_1} with the pioneering work of Bill Costerton \cite{ADN4}.

Along with HGT which plays a crucial role in the spread and the strength of antibiotic resistance, studies show that antibiotic resistance may be aided by some physical and physiological mechanisms and biofilms may also survive antibiotic treatments due to these non genetic tolerance mechanisms \cite{Cogan2006,Cogan2007,DeLeenheer2009,Cogan2005}. Physical tolerance comes from the extracellular polymeric matrix that surrounds bacterial populations. Biofilms are spatially organized, the thickness of biofilm makes the physical interaction of  antibiotic cells with bacteria difficult, and  disinfectants disperse rather than mix uniformly with bacteria \cite{Stewart1996,Stewart2003,Stewart2009}. Physiological tolerance is the reduced vulnerability of slow-growing bacteria (persistors) in nutrient-depleted environments. Antibiotics are particularly effective against fast dividing cells \cite{Tuomanen1986}, which are typically found near the biofilm surface with maximal nutrient access \cite{Xu2000}. As antimicrobials kill the outer cells, deeper bacteria may be exposed to nutrients and therapy. Insufficient treatment may destroy only the susceptible cells, leaving the resistant bacteria to survive \cite{Stewart2001}.
Phenotypic tolerance is a temporary and non-genetic tolerance, it arises when some bacteria temporarily adopt a non-dividing or stress-resistant phenotype. Two main hypotheses explain this: adaptive responses and persister formation. In adaptive responses, bacteria alter their metabolism in response to sublethal antimicrobial levels, transitioning into a more tolerant state. Mechanisms such as stress responses, detoxification, and metal resistance may play a role in adaptive responses \cite{Mangalappalli2006,Harrison2007}. While persister dynamics results in the formation of physiologically dormant cells (persisters) that withstand the antibiotic exposure irrespective of its concentration. Experiments \cite{Desai1998} show that they can survive high-dose antibiotic exposure, depending on their growth phase, and do not revert to susceptibility until the antimicrobial is removed, after which they come back to the susceptible state and continue to grow. Although adaptive responses and persistence are distinct biologically, but they can be dealt by same mathematical modeling techniques \cite{Lewis2001, Keren2004}.

Since resistant bacteria have genetic tolerance mechanisms and persister bacteria can enter a dormant state that allows them to withstand antibiotic exposure, the choice of dosage technique remains important in their population dynamics \cite{Stewart2003,Stewart2009,Xu2000,B4}. A continuous dosing may create a selection pressure, rapidly removing vulnerable cells and allowing resistant forms to dominate. In contrast, periodic dosing creates drug-free intervals, allowing persister bacteria to become metabolically active and to convert into susceptibles, during off-periods, susceptible bacteria grow faster than resistant bacteria and compete with resistant populations for space and available food, this rivalry can restrict the growth of the resistant population, particularly if resistance has a fitness cost in drug-free circumstances \cite{Cogan2006,Cogan2007,B7,B8,Cogan2005,DeLeenheer2009,B11}. Periodic dosing promotes eradication of tolerant cells and may postpone or inhibit resistance supremacy, making it a more robust long-term strategy \cite{Cogan2006,DeLeenheer2009,Hermsen2012}.

The use of mathematical models to study antibiotic resistance and successful treatment protocols has developed into an area of active research. These models are useful in understanding the complex dynamics of resistance acquisition and help in evaluating strategies to control spread of antibiotic resistance. 
Several studies have proposed and analyzed mathematical models of biofilm formation and growth including their interaction with chemical species like nutrients and antibiotics. Bottomley et al. \cite{Bottomley1} devised a one dimensional coupled ODE-PDE model for the growth of biofilm in a porous media and studied the effects of attachment-detachment and longitudinal diffusion. D’Acunto et al. \cite{Acunto2} developed a continuum-based biofilm model for multispecies based on nonlinear parabolic PDEs for nutrients, and hyperbolic PDEs for bacteria, and an ODE for biofilm length and studied the consistency and invasion dynamics. Klapper \cite{Klapper3} proposed a generalized Nash equilibrium PDE model and described biofilm’s metabolism and its interaction with the environment. Zhao et al. \cite{Zhao4} investigated  3-dimensional multiphase heterogeneous biofilms with susceptible and persister cells to analyze the effects of available concentration of nutrients on biofilms and dosing protocols for biofilm tolerance and treatment. A numerical study was done by Zhao et al. \cite{Zhao5} to solve a multispecies 3-dimentional hydrodynamic model to treat oral biofilms and showed that young biofilms are more sensitive to antibiotics, and quorum sensing and limiting persister formation help treatment. Szomolay et al. \cite{Szomolay2005} established a fundamental reaction–diffusion model that illustrated how deeper biofilm layers might become adaptively resistant to antimicrobial agents because of limited penetration, making them more difficult to eliminate as compared to planktonic cells. Szomolay \cite{Szomolay2008} extended this work by considering a moving boundary, and examined the existence and uniqueness of steady states of the model for small antibiotic doses. Further, Cogan et al. \cite{Szomolay2012} built a one-dimensional model of biofilm consisting of persister and susceptible subpopulations and examined how the periodic dosing affect the survival of persisters. They also demonstrated the critical significance of the period and dosing ratio, showing that intermittent (on-off) dosing is more effective than continuous dosing for low antimicrobial concentrations. 
 
A substantial body of literature now exists on modeling HGT and effective dosing protocols based on planktonic bacteria.  Some studies in htis domain have recommended a hit-hard and fast strategy involving few high doses of the antibiotic \cite{ADN5, ADN6, ADN7}. Hoyle et al. in their work conclude that a single large dose is never optimal, while the optimal dosing when that minimizes the total antibiotic quantity as well as leading to bacterial elimination follows a tapering pattern \cite{ADN8}; and in a follow-up work they find similar results and also validate their findings using biological experiments \cite{ADN9}. Khan et.al modeled HGT and studied optimal dosing, they find that a tapered periodic dose is optimal \cite{ADN10, ADN11}. Work by Cogan \cite{Cogan2010} support periodic disinfection which can optimize bacterial clearance and delay resistance-driven treatment failure.

Our work extends the prior models presented in \cite{Szomolay2005, Szomolay2008, Szomolay2010, Szomolay2012} by incorporating resistant bacteria into the framework, which leads to a more complex control problem and better mimics clinical realities. We investigate a one-dimensional PDE based biofilm model with susceptible, persister and resistant bacteria, we also consider the dynamics of nutrients and biocide concentrations and changing biofilm thickness. We derive some analytic results, followed by numerical simulations to investigate the effects of different dosing durations, application ratios and parameter sensitivities. Using an extended version of Pontryagin’s Maximum Principle (PMP) we determine the optimal dosing strategy, this minimizes the antibiotic dosage over time while ensuring bacterial elimination. We investigate both continuous and intermittent dosage regimens to see how the time and intensity of antibiotic application influence treatment outcomes. PMP is a well-established tool in control theory \cite{PMP1,PMP2,PMP3,PMP4}, but its application to multi-type bacterial populations inside spatially structured biofilms is limited in the literature. 

The structure of the rest of this paper is outlined below. We begin by presenting the mathematical model in Section~\ref{sec:mathematical_modeling}. Section~\ref{sec:steady_state_solutions} consists of  analytical results, including an analysis of steady-states of the system and the outcomes of periodic disinfection. These theoretical findings are complemented by numerical simulations in Section~\ref{sec:numerical_results} where fixed continuous and periodic dosing is discussed. We use Pontryagin's maximum principle to implement the optimal dosing strategies to the model in Section~\ref{sec:optimal_control}. The paper concludes in Section~\ref{sec:conclusions} with a discussion of the findings and their broader implications.

\section{Mathematical Modeling of Biofilm}
\label{sec:mathematical_modeling}
We consider a one dimensional PDE model for a biofilm, taking into account susceptible $B_s$, persister $B_p$, resistant $B_r$ and dead bacteria $B_d$. We also model the dynamics of nutrient concentration $C$ which helps in microbial growth, and the dynamics of antibiotic $A$ that acts as a bactericidal agent.  The biofilm is considered to be a continuum with the domain length given by $L$. 
If we consider the biofilm as a fully saturated continuum and neglect the fluid component, we will have:
\begin{equation}
B_s + B_p + B_r + B_d = 1
\label{eq:sum_bi}
\end{equation}

The time domain is \( t \in [0,T] \) where $T$ is fixed and the spatial domain is a moving boundary taken as \( x \in [0, L(t)]\) where $L(t)$ is the length of biofilm and is a function of $t$. It may keep on growing if no treatment is applied, converge to zero in the case of successful treatment, or converge to some non-zero state if insufficient amount of dosing is applied. Further, we assume that the boundary $x=0$ is the impermeable substratum to which the biofilm is attached and x=L denotes the interface between biofilm and bulk fluid.

To write the equations for the cell volume fractions of bacteria, we assume that the densities of the different phenotypic bacteria are constant. Thus, the laws of conservation of mass become equivalent to that of conservation of volume \cite{B8} and we have:

\textbf{Susceptible Bacteria:}
\begin{align}
\frac{\partial B_s}{\partial t} + \frac{\partial (v B_s)}{\partial x} &= (\underbrace{1}_{\text{growth}}-\underbrace{k_d A}_{\text{disinfection}} -\underbrace{k_l}_{\text{loss}}) f_s(C) B_s \nonumber \\
&\quad + \underbrace{k_r(A)B_p}_{\text{reversion of persisters}}
- \underbrace{\mu B_r B_s}_{\text{conjugation}} + \underbrace{q f_r(C)B_r}_{\text{reversion of resisters}}
\label{eq:Bs}
\end{align}

The susceptible cells in Eq.~\eqref{eq:Bs} increase due to their growth under feasible conditions and reversion back of persister and resister cells to susceptibles, and decrease because of their conversion to persister cells, conversion via conjugation (HGT) to resister cells and disinfection by antibiotic.    
In the absence of antibiotic stress, reversion from persistant and resistant phases frequently takes place, as observed in studies such as \cite{Lewis2001, DeLeenheer2009}.\\
\textbf{Persister Bacteria:}    
\begin{equation}
\frac{\partial B_p}{\partial t} + \frac{\partial (v B_p)}{\partial x} = k_l f_s(C) B_s - k_r(A)B_p - g_p(A)B_p
\label{eq:Bp}
\end{equation}
The dynamics of persister cells, governed by Eq.~\eqref{eq:Bp}, depend on the bidirectional exchange with susceptible bacteria, along with possible loss under antibiotic exposure.
By keeping the parameters positive, the function $g_p(A)$ is strictly positive, locally Lipschitz and non-decreasing, it's functional form is given by
\[
g_p(A) =\frac{k_d^p A}{L'+A}
\]
Reversion of persister bacteria to susceptible bacteria is usually considerable in the absence of antibiotics. 
Generally, it is taken as
\[
k_r(A) =
\begin{cases}
k_r^0, & \text{if } A = 0,\\[2mm]
k_r^1, & \text{if } A \ne 0,
\end{cases}
\]
where $k_r^0 \gg k_r^1$

\textbf{Resistant Bacteria:}
\begin{equation}
\frac{\partial B_r}{\partial t} + \frac{\partial (v B_r)}{\partial x} = (1-q-k_d^r A) f_r(C) B_r + \mu B_r B_s
\label{eq:Br}
\end{equation}
The resistant bacterial population $B_r$, described by Eq.~\eqref{eq:Br}, is governed by horizontal gene transfer from susceptible cells, selective pressure from the antibiotic, and reversion back to the susceptible phenotype at rate $q$ by mis-segregation.

\textbf{Dead Bacteria:}
\begin{equation}
\frac{\partial B_d}{\partial t} + \frac{\partial (v B_d)}{\partial x} = k_d A f_s(C) B_s + k_d^r A f_r(C) B_r + g_p(A) B_p
\label{eq:Bd}
\end{equation}

\textbf{Antibiotic:}

The spatiotemporal dynamics of the antibiotic concentration \(A(x,t)\) is modeled by the following diffusion–reaction equation:
\begin{equation}
\frac{\partial A}{\partial t} = D_A \frac{\partial^2 A}{\partial x^2} - p_s(A) B_s - p_p(A) B_p - p_r(A) B_r
\label{eq:A}
\end{equation}

Here the term \(D_A \frac{\partial^2 A}{\partial x^2}\), represents the diffusion of antibiotic molecules from areas of high concentration to low concentration, where \(D_A\) is the diffusion coefficient of the antibiotic. The passive dispersion of the antibiotic molecules through the heterogeneous biofilm is processed by Fick's law. The other terms $- p_s(A) B_s - p_p(A) B_p - p_r(A) B_r$ account for the consumption of the antibiotic due to its interaction with susceptible, persister and resistant bacteria respectively.

The functions \(p_i(A), \; i = s, p, r\) may be taken as Michaelis-Menten type, Hill type or simply linear in A as \( c_i A\). We have taken them as \[p_i(A) = \frac{\nu_i A}{L_i + A} , \; i = s, p, r \]
where $\nu_s > \nu_r >\nu_p$ and $L_s \le L_r \le L_p$  which guarantee that the consumption of antibiotic by susceptible bacteria is highest and that of by persister bacteria is lowest, and $p_i(A)$ are such that they vanish with A, and are non-decreasing functions in $A$:
\[p_i(0) = 0, \quad
p_i'(A) \ge 0
\]
A more detailed description of pharmacodynamic functions can be found in \cite{ADN12} 

\textbf{Nutrient Concentration:}
\begin{equation}
\frac{\partial C}{\partial t} = D_C \frac{\partial^2 C}{\partial x^2} - f_s(C) B_s - f_r(C) B_r
\label{eq:C}
\end{equation}
The first term on the right hand side represents the diffusion of the nutrient molecules inside the biofilm and $D_C$ is the diffusion coefficient for $C$. The second and third terms represent the consumption of nutrients by susceptible and resistant bacteria respectively. Since the persister cells don't use the nutrients or may use a negligible amount of $C$ due to their dormant state, it has not been included in Eq.~\eqref{eq:C}.

It is important to note that biofilms show limited diffusion because of their increased thickness and decreased fluid  flow. As indicated by the ratio $D_e/D_{aq}$, the effective diffusivity usually falls between $10^{-6}$ and $10^{-5}$ cm$^2$/s, where $D_e$ and $D_{aq}$ represent the solute diffusion coefficients in biofilm and water, respectively. For example, $D_e/D_{aq}$ is roughly 0.25 for the majority of organic molecules and 0.6 for gases such as oxygen and methane \cite{Stewart2003}. Stewart has estimated the diffusion coefficients of antibiotics \cite{Stewart1996}.
With $f_i(0) = 0$, we regard the functions $f_i$ for $i = s,r$ as locally Lipschitz continuous and nonnegative given by
\[
f_i(C) =\frac{\alpha_i C}{L_i + C}, \quad i=s,r
 \quad \text{with} \; \alpha_s > \alpha_r\]

\textbf{Advective Velocity:}
\begin{equation}
\frac{\partial v}{\partial x} = f_s(C) B_s + f_r(C) B_r
\label{eq:v}
\end{equation}
The advective velocity $v$ is derived by summing Eqs.~\eqref{eq:Bs}-\eqref{eq:Bd} and using the incompressibility condition~\eqref{eq:sum_bi}.

\textbf{Biofilm Length:}
\begin{equation}
\frac{dL}{dt} = v(L,t) - \sigma L^2
\label{eq:L}
\end{equation}

where $\sigma$ is the detachment constant. According to one-dimensional models given in \cite{Szomolay2010, Szomolay2008}, advection and erosion cause biofilm thickness $L(t)$ to change. In contrast to sloughing, which involves the separation of larger fragments, erosion is described as the gradual removal of small biomass clusters from the surface \cite{B4}. There are several different detachment models \cite{Kommedal2003}. Here, we employ a model where detachment is proportional to bacterial cells, and it depends quadratically on biofilm thickness.

The initial and boundary conditions for the system are as 
    \[ B_i(x, 0) = B_{i_{0}}, \qquad
    \left. \frac{\partial B_i(x, t)}{\partial x} \right|_{x = 0} = 0, \qquad \text{where} \; i =s, p, r, d\]
    \[ \left. \frac{\partial A(x, t)}{\partial x} \right|_{x = 0} = 0,     \qquad
A(L(t), t) = u(t) \quad \text{(Control function)}
\]
\[
\left. \frac{\partial C(x, t)}{\partial x} \right|_{x = 0} = 0,     \qquad
C(L(t), t) = K
\]
\[
v(0,t)=0, \qquad L(0)=L_0
\]

The one-dimensional PDE system, therefore, accounts for interactions among three bacterial subpopulations — susceptible ($B_s$), persister ($B_p$), and resistant ($B_r$) — along with nutrient and antibiotic dynamics. This formulation extends the two-state biofilm models studied by Cogan et al. \cite{Szomolay2012} and De Leenheer et al. \cite{DeLeenheer2009} to include antibiotic-driven resistance mechanisms, bacterial conjugation and will be used to investigate optimal control of antibiotic treatment to eradicate the bacteria. 

Adaptive resistance and persister-mediated tolerance produce mathematically similar structures, despite being biologically distinct processes. In contrast to persisters which remain metabolically inert, resistant cells continue to proliferate partially during antibiotic exposure. As a result, transitions between these states shape the overall population behavior and therapeutic outcome.

In our model, we consider two types of time dependent antibiotic dosing; continuous and periodic. In practical applications, the antibiotic administration is usually represented as a sequence of doses applied after some fixed period P recommended by doctors and researchers, while mathematically is can be approximated by Gaussian pulse train function given as

\begin{equation}
u(t) = \sum_i A_0 \, \exp(-a(t - iT)^2),
\end{equation}
The parameter $a>0$ controls the width (spread) of the pulse, $A_0$ denotes the dose magnitude, and $T$ the dosing period.

\section{Steady State Solutions}
\label{sec:steady_state_solutions}
Steady states are important for the study of biological models as they describe the long term behaviour of the populations, and determine whether the populations die out, persist or coexist. We analyze all possible steady-states of the model in this part. We will prove the sufficient conditions for the presence of both trivial and non-trivial stable states. The results provided here will be compared with those of a related analytical study of steady states for a biofilm model by Cogan et al. \cite{Szomolay2012}.

The steady state equations for the system ~\eqref{eq:Bs} -~\eqref{eq:L}  are 

\begin{align}
 \frac{d (v B_s)}{d x} &= (1-k_d A -k_l) f_s(C) B_s + k_r(A)B_p
-\mu B_r B_s +q f_r(C)B_r
\label{eq:Bs_s1}
\end{align}

\begin{equation}
\frac{d (v B_p)}{d x} = k_l f_s(C) B_s - k_r(A)B_p - g_p(A)B_p
\label{eq:Bp_s1}
\end{equation}

\begin{equation}
\frac{d (v B_r)}{d x} = (1-k_d^r A - q) f_r(C) B_r + \mu B_r B_s
\label{eq:Br_s1}
\end{equation}

\begin{equation}
\frac{d (v B_d)}{d x} = k_d A f_s(C) B_s + k_d^r A f_r(C) B_r + g_p(A) B_p
\label{eq:Bd_s1}
\end{equation}

\begin{equation}
D_A \frac{d^2 A}{d x^2} = p_s(A) B_s + p_p(A) B_p + p_r(A) B_r
\label{eq:A_s1}
\end{equation}

\begin{equation}
 D_C \frac{d^2 C}{d x^2} = f_s(C) B_s + f_r(C) B_r
\label{eq:C_s1}
\end{equation}

\begin{equation}
\frac{d v}{d x} = f_s(C) B_s + f_r(C) B_r
\label{eq:v_s1}
\end{equation}

The corresponding boundary conditions will be as 
    \[ \frac{d B_i(0)}{d x} = 0, \qquad where \qquad i =s, p, r, d\]
    
    \[ \frac{d A(0)}{d x} = 0,     \qquad
A(L) = U
\]
    \[ \frac{d C(0)}{d x} = 0,     \qquad
C(L) = K
\]
\[
v(0)=0, \qquad v(L) = \sigma L^2
\]

We can note that $B_d$ can be uniquely expressed in terms of $B_s$, $B_p$ and $B_r$, thus we may reduce our system to the following three equations:

\begin{equation}
v \frac{d B_s}{d x} = (1-k_d A - k_l -B_s) f_s(C) B_s +k_r(A) B_p -(\mu + f_r(C))B_r B_s + q f_r(C) B_r 
\label{eq:Bs_s2}
\end{equation}

\begin{equation}
v \frac{d B_p}{d x} = (k_l-B_p) f_s(C) B_s -k_r(A) B_p -g_p(A)B_p -f_r(C) B_r B_p
\label{eq:Bp_s2}
\end{equation}

\begin{equation}
v \frac{d B_r}{d x} = (1-k_d^r A - q -B_r) f_r(C) B_r + (\mu-f_r(C)) B_r B_s
\label{eq:Br_s2}
\end{equation}

The reparameterization of above system ~\eqref{eq:Bs_s2} - ~\eqref{eq:Br_s2} using the following equation of characteristics may help to analyze the steady states.

\begin{equation}
\frac{d s}{d t} =  v (s(t)), \qquad s(0) = L
\label{eq:Eq_Char}
\end{equation}
where the parameter $t$ is selected so that $s(t)$ remains in the interval $(0,L]$.
Eq.~\eqref{eq:Eq_Char} has a $C_1$-solution on \((-\infty, 0]\) even when the biofilm is well embedded and only dead cells are left,
 i.e. $B_s + B_p + B_r =0$ on \([0,x_0]\) for some $0\le x_0\le L$ (for details see the work by Cogan et al. \cite{Szomolay2012}).
 \\
The reparameterization of steady-state solutions will be used from now on. Define
\[ b_s(t) := B_s(s(t)), \qquad  b_p(t) := B_p(s(t))\]
\[ b_r(t) := B_r(s(t)), \qquad  a(t) := A(s(t))\]
\[ c(t) := C(s(t)).\]

Thus, by the parametrization ~\eqref{eq:Eq_Char}, we can write our system of equations ~\eqref{eq:Bs_s2}-~\eqref{eq:Br_s2} as:

\begin{align}
\frac{d b_s}{d t} = (1-k_d a(t) - k_l -b_s) f_s(c(t)) b_s +k_r(a(t)) b_p -(\mu + f_r(c(t)))b_r b_s + q f_r(c(t)) b_r 
\label{eq:bs}
\end{align}

\begin{equation}
\frac{d b_p}{d t} = (k_l-b_p) f_s(c(t)) b_s -k_r(a(t)) b_p -g_p(a(t))b_p -f_r(c(t)) b_r b_p
\label{eq:bp}
\end{equation}

\begin{equation}
\frac{d b_r}{d t} = (1-k_d^r a(t) - q -b_r) f_r(c(t)) b_r + (\mu-f_r(c(t))) b_r b_s
\label{eq:br}
\end{equation}

System~\eqref{eq:bs}--\eqref{eq:br} does not depend on $v$. By the continuity of $s$, $a$, and $c$, we have 
\[
a(t) \to A_0 \quad \text{and} \quad c(t) \to C_0 \quad \text{as } t \to -\infty,
\] 
where $A_0 = A(x_0)$ and $C_0 = C(x_0)$. Hence, in the limit $t \to -\infty$, the non-autonomous system~\eqref{eq:bs}--\eqref{eq:br} can be approximated by the autonomous system~\eqref{eq:bs_s1}--\eqref{eq:br_s1}, which we study next.

\subsection{Autonomous System}

\begin{align}
\frac{d b_s}{d t} = (1-k_d A_0 - k_l -b_s) f_s(C_0) b_s +k_r(A_0) b_p -(\mu + f_r(C_0))b_r b_s + q f_r(C_0) b_r 
\label{eq:bs_s1}
\end{align}

\begin{equation}
\frac{d b_p}{d t} = (k_l-b_p) f_s(C_0) b_s -k_r(A_0) b_p -g_p(A_0)b_p -f_r(C_0) b_r b_p
\label{eq:bp_s1}
\end{equation}

\begin{equation}
\frac{d b_r}{d t} = (1-k_d^r A_0 - q -b_r) f_r(C_0) b_r + (\mu-f_r(C_0)) b_r b_s
\label{eq:br_s1}
\end{equation}

It is important to note that the range of solutions for the Eq.~\eqref{eq:bs_s1}-~\eqref{eq:br_s1} is a triangle lying in the first quadrant and is given by
\[
\Delta = \left\{ (b_s, b_p, b_r) \in \mathbb{R}^3 \;:\; b_s \ge 0,\; b_p \ge 0,\; b_r \ge 0,\; b_s + b_p + b_r \le 1 \right\}
\]
and that $\Delta$  is a positively invariant region.

Dropping the derivatives of $b_i$ in the autonomous system~\eqref{eq:bs_s1}-~\eqref{eq:br_s1}, we note that it has a trivial equilibrium $(0,0,0)$.
Also, taking $b_r=0$ from equation~\eqref{eq:br_s1} and putting it in equations~\eqref{eq:bs_s1}-~\eqref{eq:bp_s1}, we get
\begin{equation}
b_s^e = \frac{\sqrt{D} - g_p - k_r + f_s(m_1-k_l)}{2 f_s}
\label{eq:bse}
\end{equation}

\begin{equation}
b_p^e = \frac{ k_l \,\bigl(\sqrt{D} - g_p  - k_r + f_s (m_1 - k_l)\bigr) }
       { \sqrt{D} + g_p + k_r + f_s(m_1-k_l) }.
\label{eq:bpe}
\end{equation}

where \( D=  (g_p + k_r + f_s (k_l - m_1))^2 + 4 f_s ((g_p + k_r) m_1- g_p k_l) > 0,
\) and we denote $g_p(A_0) = g_p$ and $f_s(C_0)=f_s$.     Thus, $(b_s^e, b_p^e,0)$ is a non trivial steady state for our model.

Also, we have another non-trivial steady state $(b_s^a, b_p^a, b_r^a)$ that can be obtained by getting $b_r^a$, a non-zero value of $b_r$ from equation~\eqref{eq:br_s1} and then substituting it in equations~\eqref{eq:bs_s1}-~\eqref{eq:bp_s1}. The analytic calculations are omitted from the paper as the expressions get unwieldily, but we study this behavior numerically in section~\ref{sec:numerical_results}. 
Thus, we have three biologically possible steady states $\boldsymbol{S_1}$ = $(0,0,0)$, $\boldsymbol{S_2}$ = $(b_s^e, b_p^e,0)$, $\boldsymbol{S_3}$ = $(b_s^a, b_p^a,b_r^a)$ for our system.

The following lemmas and theorems discuss the nature of these three steady states in detail and give certain conditions for their existence and stability.

\begin{lemma}
Define 
\(
m_1 := 1 - k_d A_0, 
\; m_2 := 1 - k_d^r A_0, \;
m_1^\ast := \frac{g_p\,k_l}{\,g_p + k_r\,}
\)
and let $\lambda_1,\lambda_2,\lambda_3$ are the eigen values of the autonomous system~\eqref{eq:Bs_s2}-~\eqref{eq:Br_s2} obtained by linearizing it at the trivial equilibrium 
\((b_s,b_p,b_r)=(0,0,0)\),  then
\begin{enumerate}
    \item $\lambda_1 < 0$ for all the parameter values allowed biologically.
    \item $\lambda_2>0 \iff m_1 > m_1^\ast$.
    \item $\lambda_3>0 \iff m_2>q$.
\end{enumerate}
\label{lem:lemma1}
\end{lemma}

\begin{lemma}
The system admits at least one nonzero steady state if 
\begin{enumerate}
    \item the domain is convex, compact and positively invariant.
    \item either $m_1>m_1^*$ or $m_2>q$.
\end{enumerate}

\noindent\textbf{Trivial case.}  
If $m_1\le m_1^*$ and $m_2\le q$, then there exists only a trivial equilibrium $(0,0,0)$.
\label{lem:lemma2}
\end{lemma}

\begin{lemma}
If $m_1 > k_l$ and  $f_r (m_2 - q) < (m_1-k_l) \,(f_s - \mu)$
, then the nontrivial equilibrium 
$(b_s^e, b_p^e, 0)$ is a sink. 
Further,  if $k_r = 0=g_p$ , then the point $(m_1 - k_l, k_l, 0)$ is a stable node.
\label{lem:lemma3}
\end{lemma}

\begin{lemma}
Assume
\begin{enumerate}
    \item $\Delta$ is compact and positively invariant,
    \item $m_1<m_1^*$ and $m_2<q$,
\end{enumerate}

Then the trivial equilibrium $(0,0,0)$ is globally asymptotically stable.
\label{lem:lemma4}
\end{lemma}

Next we state Theorems~\ref{th:theorem1}-~\ref{th:theorem2} describing the existence of steady-state solutions in more rigorous way.

Note that our model always admits a steady-state solution such that $L = 0$. We want to know that is there a steady-state solution for $L > 0$? To answer this, assume that $L > 0$ is given, and we will prove the existance of a non trivial steday state  for it.

\begin{theorem}
Fix $L>0$. Let there exist $C_0\in(0,K]$ and  $A_0\in(0,U]$ such that the frozen autonomous system~\eqref{eq:bs_s1}-~\eqref{eq:br_s1} has a nontrivial solution $(b_s,b_p,b_r)\in\Delta$ (Lemma 2 guarantees it under certain conditions),

Then there exists at least one continuous, nontrivial steady state $(B_s, \allowbreak B_p, \allowbreak B_r, \allowbreak B_d,\allowbreak A, \allowbreak C, \allowbreak v)$ for the system~\eqref{eq:Bs_s1}-~\eqref{eq:v_s1} satisfying the set of corresponding boundary conditions.

Conversely, if a continuous nontrivial steady-state solution exists then its value at $x=0$  must equal to the frozen equilibrium, i.e.,
$$
\big(B_s(0),B_p(0),B_r(0)\big) = (b_s,b_p,b_r),
$$

where $(b_s,b_p,b_r)$ is the frozen equilibrium with $A_0=A(0),\,C_0=C(0)$.
\label{th:theorem1}
\end{theorem}

\begin{theorem}
For fixed $L>0$, and for
$m_1 \le m_1^* \;\text{and}\; m_2 \le q,$
there exists only trivial equilbrium of the system, i.e., 
$
B_s(x)\equiv B_p(x)\equiv B_r(x)\equiv 0, \quad \forall \; x\in[0,L].$
\label{th:theorem2}
\end{theorem}

In the previous two theorems, we did not consider the ODE of $L$ and took $L$ to be fixed. Next we will give a very useful theorem that tells us the minimum dose required to eradicate the biofilm and to have $L=0$.

\begin{theorem}
There exist \(0\le u_1\le u_2<\infty\) such that
\begin{enumerate}
  \item For $u_0 \in [0,u_1)$,  there exists at least one continuous nontrivial steady-state with \(L>0\), and
  \begin{equation}
    u_1 \;\ge\; \min\!\Big\{\,\frac{1-m_1^\ast}{k_d},\;\frac{1-q}{k_d^r}\,\Big\}  
  \label{eq:insuf_dose}
  \end{equation}
  \item For $u_0\ge u_2$, there exists only trivial steady-state solution  with $L=0$, and 
  \begin{equation}
    u_2 \le \frac{\cosh\!\big(\phi_a^*\, g_*(K)/\sigma\big)}{\min(k_d,k_d^r)}
  \label{eq:suf_dose}
  \end{equation}
  where \( (\phi_a^*)^2 := \frac{\underline\alpha}{D_A} \; \text{and} \;  g_*(K):=f_s(K)+f_r(K)\).
\end{enumerate}
\label{th:theorem3}
\end{theorem}

 The corresponding proofs of lemmas and theorems are given in the Appendix. 

\section{Numerical Results}
\label{sec:numerical_results}
We discuss numerical simulations in this section. Two different analysis have been done and presented in this section. First, we apply continuous and periodic dosing by trial-and-error approach to kill the bacteria. In this approach, the initial dose is kept constant throughout the simulations. Later we apply Pontryagin's maximum principal (PMP) to optimize the dose for both continuous and periodic scenarios which results in a tapered dosing schedule. PMP finds the amount of dosing that is required to apply at the substratum to kill all the bacteria and to eradicate the biofilm. An important part of this section is alignment of these numerical results with the lemmas and theorems discussed earlier. We calculated the total dose used in trial-and-error approach and compared them with the case when an optimal dosing strategy based on dose tapering governed by Pontryagin's maximum principle is adopted to kill the bacteria.

Since mathematical analysis shows that three steady states exist for~\eqref{eq:Bs}-~\eqref{eq:L}, hence based on this analysis we emphasize on three steady states for fixed continuous dosing (see in Figures~\ref{fig:myimage1}--\ref{fig:myimage3}) and fixed periodic dosing (see in Figures~\ref{fig:myimage4}--\ref{fig:myimage6}). After that we have added the numerical results for optimal analysis showing the profile of $\delta(t)$, where $u(t)= A_0 \delta(t)$, calculated numerically governed by the optimality condition.

\begin{table}[H]
\centering
\caption{Description of parameters used}
\vspace{1mm}
\label{tab:parameters}
\begin{tabular}{|p{5cm}|l|l|l|}
\hline
\textbf{Description} & \textbf{Symbol} & \textbf{Units} & \textbf{Value} \\
\hline
Diffusion coefficient of antibiotic & $D_A$ & $mm^2 h^{-1}$ & $7 \times 10^{-2}$ \cite{DC1}  \\
\hline
Diffusion coefficient of nutrient & $D_C$  & $mm^2 h^{-1}$ & $7 \times 10^{-2}$ \cite{DC1} \\
\hline
Killing rate of susceptible bacteria & $k_d$ &  $h^{-1}$ & 0.35  \\
\hline
Killing rate of resistant bacteria & $k_d^r$ & $h^{-1}$ & 0.15 \\
\hline
Killing rate of persister bacteria & $k_d^p$ & $h^{-1}$ & 0.1 \\
\hline
Conversion rate of $B_s$ to $B_p$ & $k_l$ & $h^{-1}$ & 0.01 \cite{Szomolay2012} \\
\hline
Reversion of persister to susceptible bacteria & $k_r$ & $h^{-1}$ & 0.005 \cite{Szomolay2012}  \\
\hline
Conjugational transfer rate & $\mu$ & $h^{-1}$ & 0.0000001 \cite{ADN10}, \cite{ADN12} \\
\hline
Removal rate of antibiotic by $B_s$ & $\nu_s$ & $h^{-1}$ & 0.345 \cite{ADN10}, \cite{ADN12} \\
\hline
Removal rate of antibiotic by $B_p$ & $\nu_p$ &$h^{-1}$ & 0.1\\
\hline
Removal rate of antibiotic by $B_r$ & $\nu_r$ &$h^{-1}$ & 0.3  \\
\hline
Maximum specific growth rate of $B_s$ & $\alpha_s$ & $h^{-1}$ & 1  \cite{Szomolay2012} \\
\hline
Maximum specific growth rate of $B_r$ & $\alpha_r$ & $h^{-1}$ & 0.5 \\
\hline
Detachment rate per unit length & $\sigma$ & $mm^{-1}\; h^{-1}$ & 0.1 \cite{Szomolay2012} \\
\hline
Fractional segregation loss parameter  & $q$ & dimensionless & 0.01 \cite{ADN10}, \cite{ADN12}\\
\hline
Initial dimensionless density of susceptible bacteria& $B_{s_0}$ & dimensionless & 0.70 \\
\hline
Initial dimensionless density of persistent bacteria& $B_{p_0}$ & dimensionless & 0.20 \\
\hline
Initial dimensionless density of resistant bacteria& $B_{r_0}$ & dimensionless & 0.10  \\
\hline
Initial dimensionless density of dead bacteria & $B_{d_0}$ & dimensionless & 0  \\
\hline
Initial dimensionless density of nutrient& K  & dimensionless & $1$ \cite{Szomolay2012} \\
\hline
Half saturation constant & $L_i$ & dimensionless & 0.1   \cite{ADN10}, \cite{ADN12} \\
\hline
Antibiotic sensitivity parameter & $\beta$ & dimensionless & 0.001  \cite{ADN10}, \cite{ADN12} \\
\hline
Treatment time & T & $ h $& 100 \cite{Szomolay2012} \\
\hline
Initial biofilm length & $L_0$ & mm& 5 \cite{Szomolay2012} \\
\hline
Width of the pulse  of antibiotic dose & a & $[h^{-1}]$ & 5 \\
\hline
\end{tabular}
\end{table}

\subsection{Continuous Dosing}
Figure~\ref{fig:myimage1} shows the full non trivial steady state in which all the bacteria and biofilm length $b_s, b_p, b_r, L \neq 0$ and leads to treatment failure when the initial dose is insufficient to kill the bacteria. In this case, the administered dose is 1 and the condition for the existance of non trivial steady state is satisfied: $m_1=0.65 > m_1^*=0.01$ and $m_2=0.85 > q=0.01$

Figure~\ref{fig:myimage2} shows the second non trivial equilibrium $b_s, b_p, L \neq 0,$ and $b_r=0$, this is obtained by  increasing the initial dose to 2.5 by keeping the parameter values same, where $m_1=0.125 > m_1^*=0.01$ and $m_2=0.625 > q=0.01$. 

Finally, when we increase the initial dosing to 10, we get the trivial steady state $b_s=0, b_p=0, b_r=0, L=0$ that can be seen in Figure~\ref{fig:myimage3}. In this case, $m_1=-2.5 \le m_1^*=0.01$ and $m_2=-0.5 \le q=0.01$. These results are consistent with our analytic results stated in lemmas and theorems. The total antibiotic used over the entire time course in these three cases is 100, 250 and 1000 units respectively.

\begin{figure}[H]
    \centering
    \includegraphics[width=0.7\textwidth]{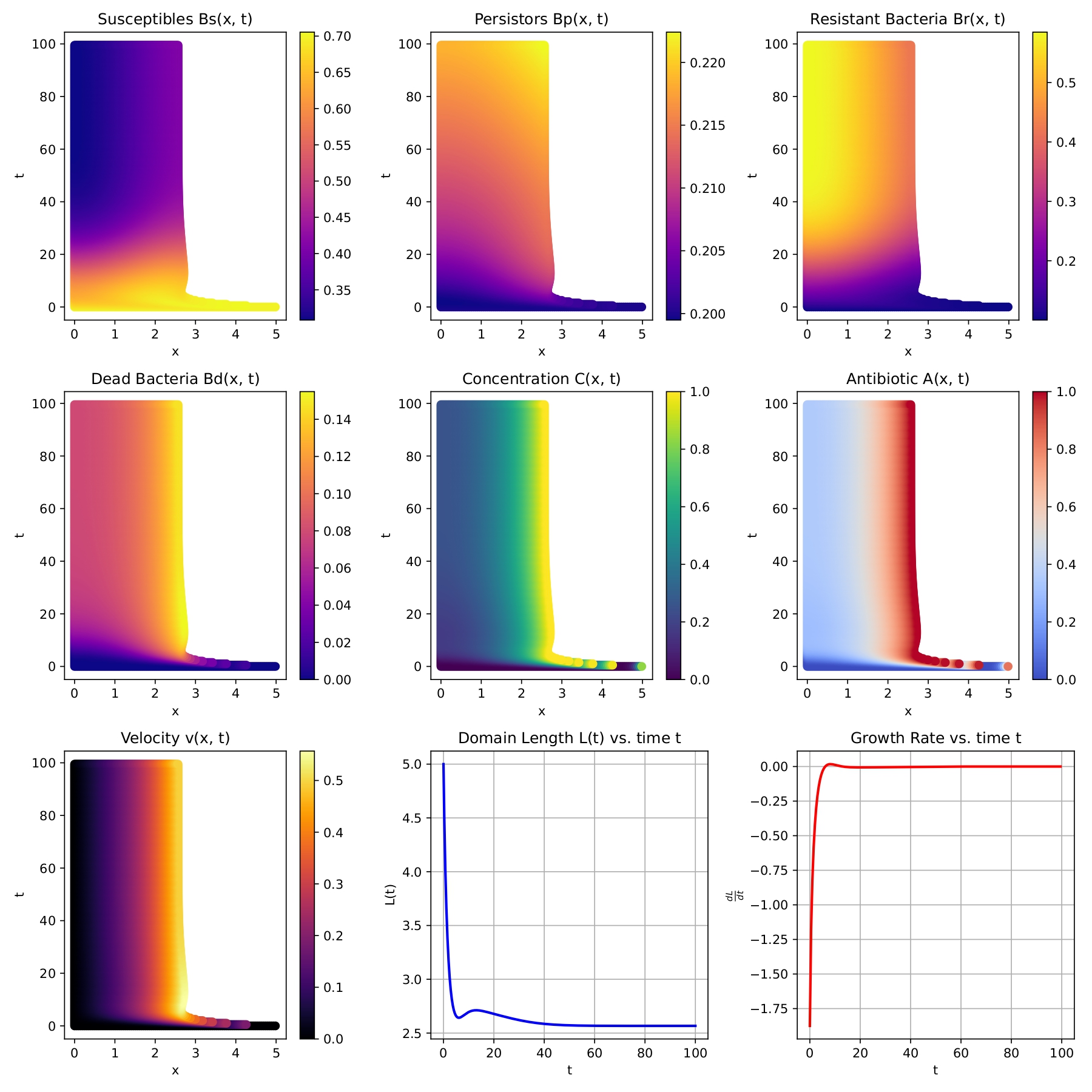} 
    \caption{ Completely Non-Trivial Steady State (Total antibiotic used: 100 units)}
    \label{fig:myimage1}
\end{figure}

\begin{figure}[H]
    \centering
    \includegraphics[width=0.7\textwidth]{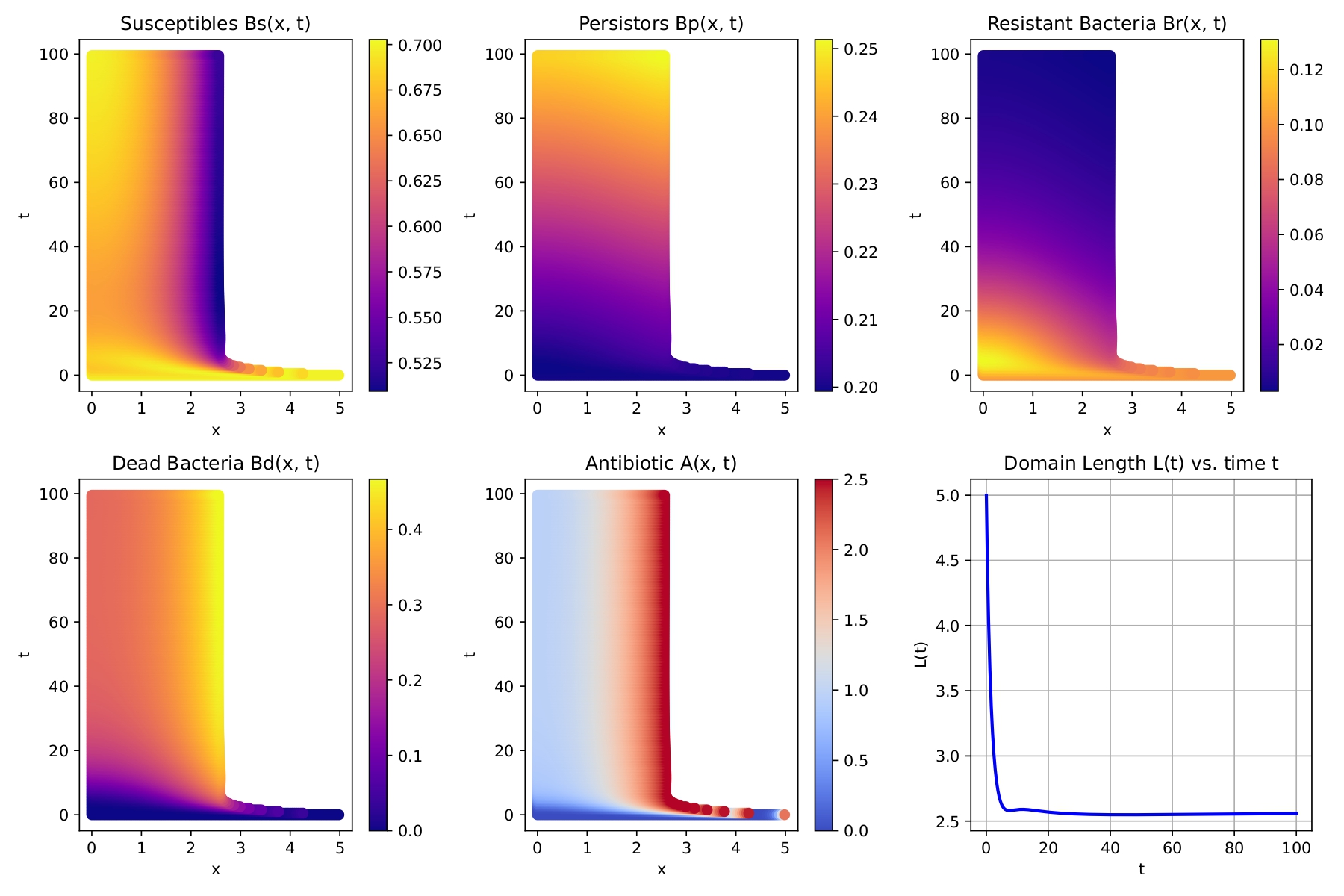} 
    \caption{Partially Non-Trivial Steady State (Total antibiotic used: 250 units)}
    \label{fig:myimage2}
\end{figure}

\begin{figure}[H]
    \centering
    \includegraphics[width=0.7\textwidth]{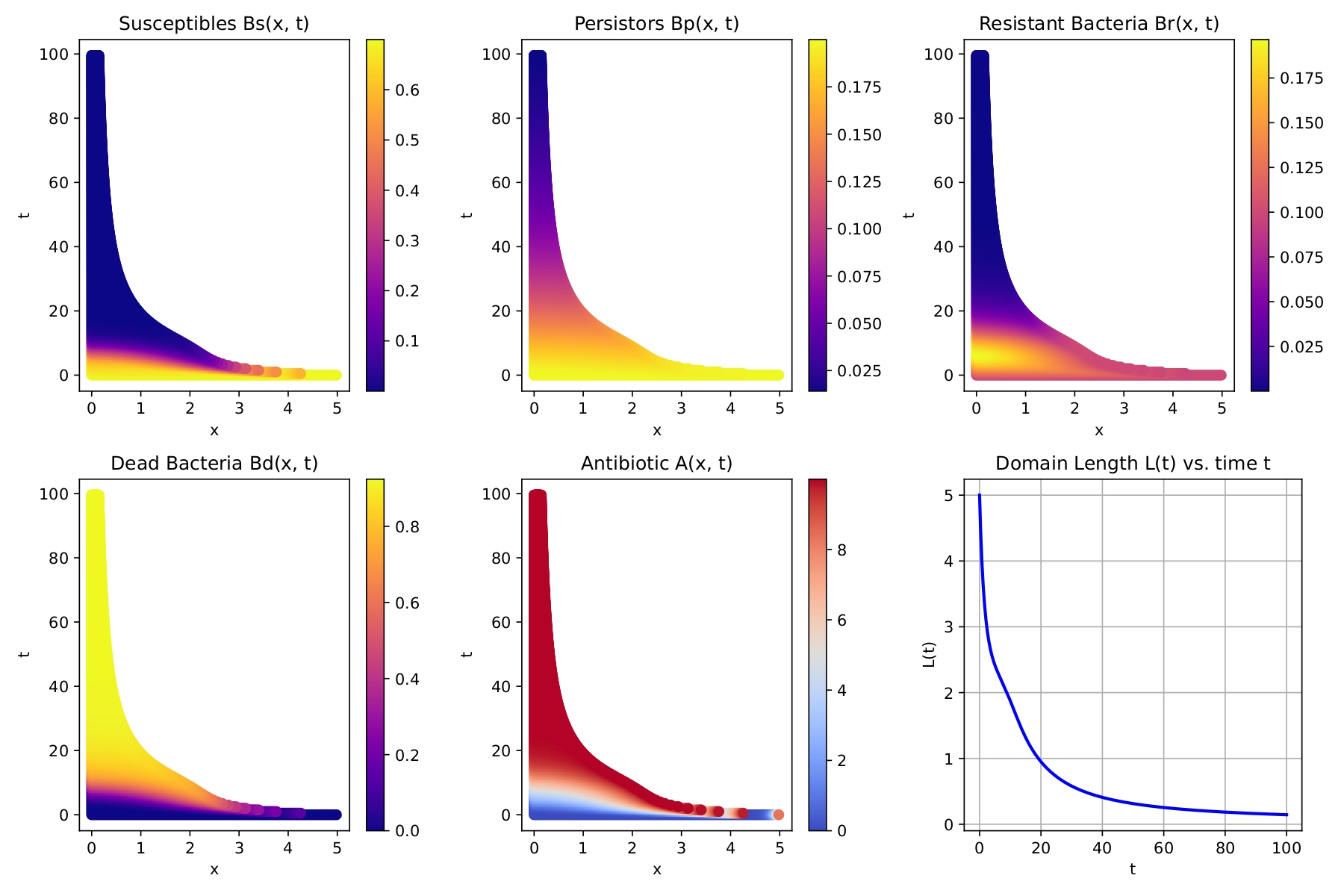}
    \caption{Trivial Steady State (Total antibiotic used: 1000 units)}
    \label{fig:myimage3}
\end{figure}
\subsection{Periodic Dosing}
 
Figure~\ref{fig:myimage4} illustrates the complete non-trivial steady state where all bacterial populations and the biofilm length remain positive, i.e., $b_s, b_p, b_r, L \neq 0$, when dose is applied to be $5$ units after each 6 hours. This scenario represents treatment failure, occurring when the administered dose is insufficient to eradicate the bacteria.
Figure~\ref{fig:myimage5}  presents the second non-trivial equilibrium, where $b_s, b_p, L \neq 0$ but $b_r=0$. This state emerges when the initial dose is increased to $8$ units, while keeping the parameters unchanged.
Finally, by further increase in the dose to $40$ units, the system reaches the trivial steady state with $b_s=0, b_p=0, b_r=0, L=0$ that is shown in Figure~\ref{fig:myimage6}. These outcomes are consistent with the analytical predictions established in the lemmas and theorems. Note that the antibiotic consumed over the entire time course in these three periodic dosing scenarios is 80, 128 and 640 units respectively.

In the continuous dosing strategy, a total antibiotic amount of approximately 1000 units was required to achieve complete biofilm eradication. In contrast, the  periodic dosing approach achieved successful treatment with only 640 units of the drug, indicating a more efficient use of antibiotics when dosing is applied intermittently. This clearly demonstrates that the periodic administration pattern enhances treatment effectiveness by allowing bacterial populations to experience varying drug concentrations over time, which reduces the persistence and resistance effects.
\begin{figure}[H]
    \centering
    \includegraphics[width=0.7\textwidth]{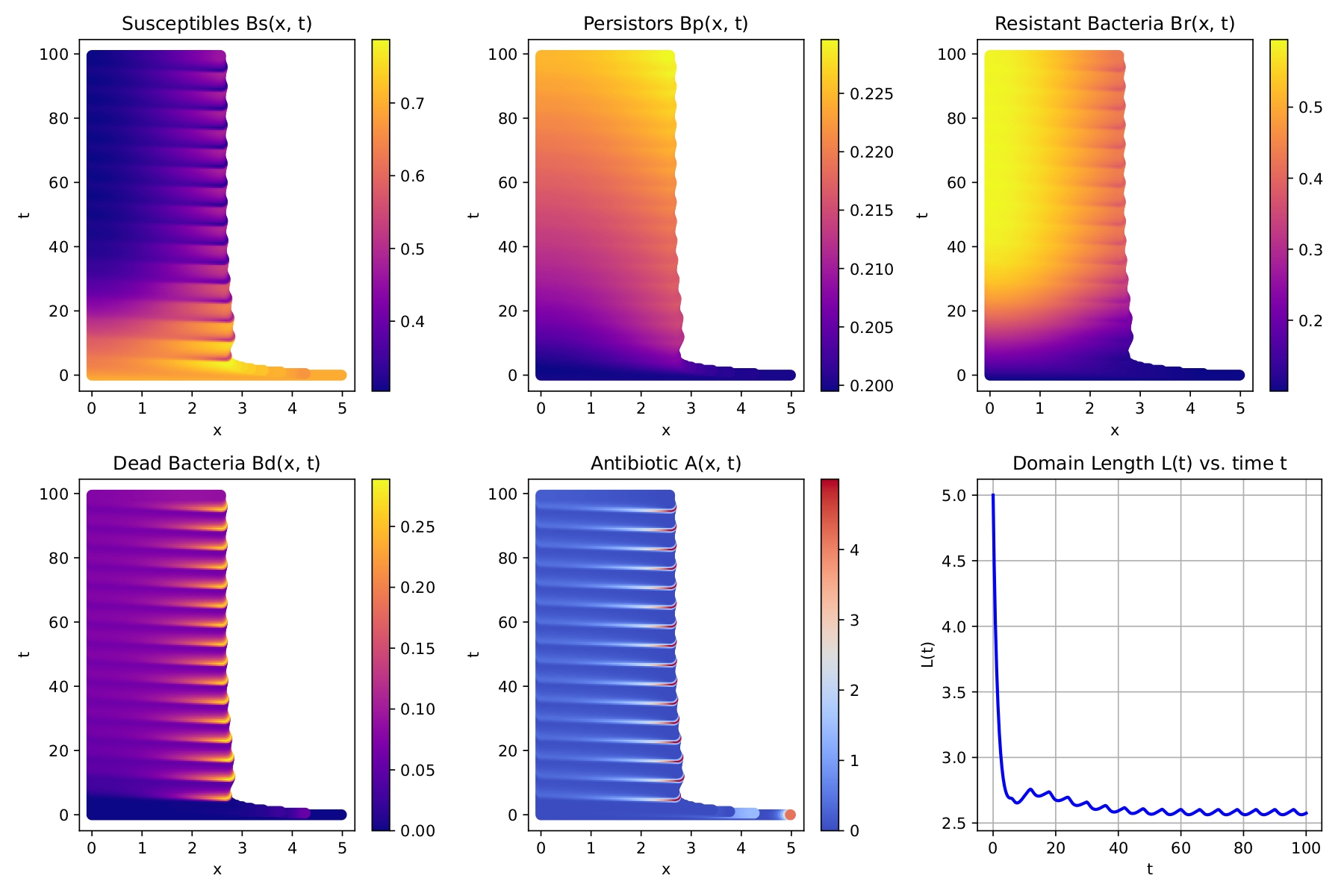} 
    \caption{Completely Non-Trivial Steady State (Total antibiotic used: 80 units)}
    \label{fig:myimage4}
\end{figure}

\begin{figure}[H]
    \centering
    \includegraphics[width=0.7\textwidth]{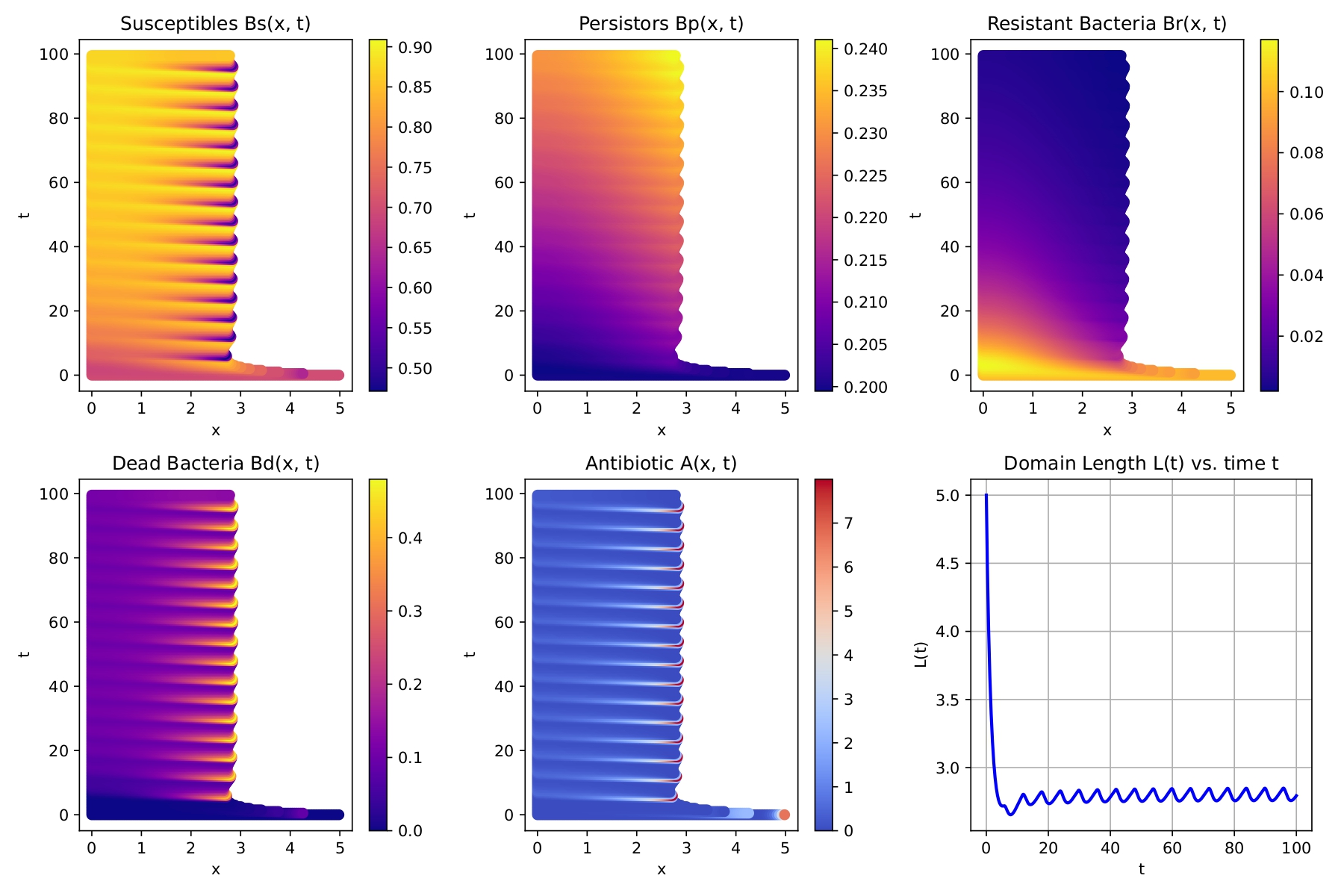}
    \caption{Partially Non-Trivial Steady State (Total antibiotic used: 128 units)}
    \label{fig:myimage5}
\end{figure}

\begin{figure}[H]
    \centering
    \includegraphics[width=0.7\textwidth]{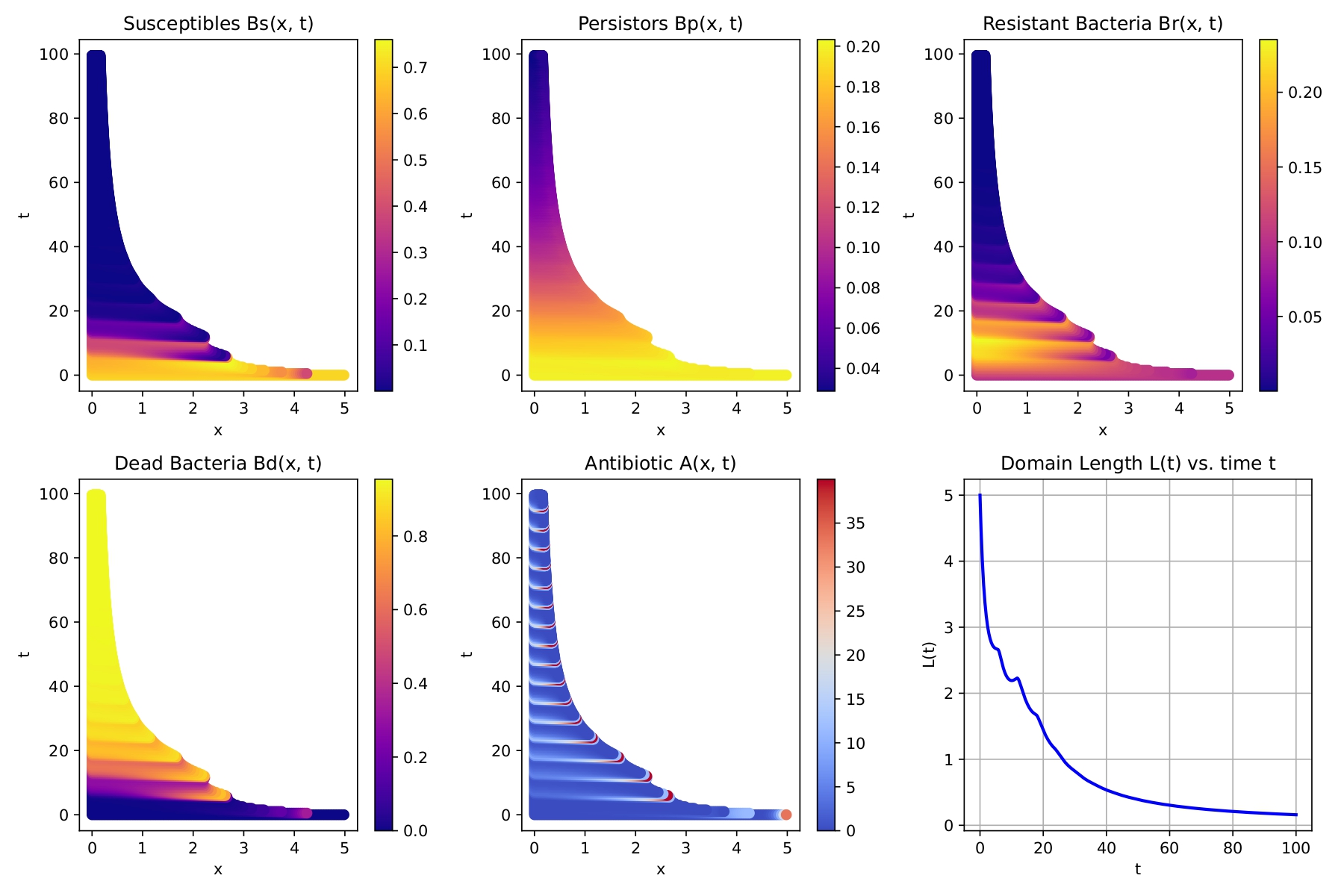} 
    \caption{Trivial Steady State (Total antibiotic used: 640 units)}
    \label{fig:myimage6}
\end{figure}

We note in passing that the seemingly paradoxical result of the existence of the steady state $b_r=0, b_s,b_p \neq 0 $ in which the susceptible bacteria survive (along with persisters) while resistant bacteria are eradicated, is a consequence of mis-segregation of resistant bacteria. As long as there are resistant bacteria, they will give rise to susceptible ones during cell division, hence it is impossible to have a steady state with no susceptible bacteria while the resistant strain persists. 

\section{Optimal Control}
\label{sec:optimal_control}
This section addresses the challenge of identifying an antibiotic treatment plan that minimizes the overall amount of antibiotic used while also killing the active bacteria. Numerous studies have shown that using antibiotics excessively can have unfavorable effects \cite{OC2}. In fact, it has also been proposed that this could potentially make people more susceptible to infection by boosting their effective resistance \cite{OC1}. We are also motivated to use as few antibiotics as possible during the course of treatment due to their high price.

Since the dilution rate is too low to completely destroy the bacteria on their own, antibiotics must be administered concurrently. However, since no new and superior antibiotics have been developed to combat microbial strains, we must use the ones we do have more efficiently \cite{OC3}.
Our effort focuses on developing an effective antibiotic application approach that ensures minimal drug deployment while simultaneously eliminating germs.
As stated in \cite{OC4}, we propose the best course of action for our model using optimal control theory.
The results indicate that the best course of therapy is to gradually decrease the strength of the administered doses while alternating between the application and withdrawal of antibiotics.

On the basis of this, both provide appropriate dosing techniques for beginning conditions and parameter values. We have observed that these strategies prevent the bi-stability that arises in the non-optimal antibiotic application situation i.e., the bacteria-free state is stable for a wide range of starting bacterial concentrations. Therefore, according to our hypothesis, using such a course of treatment guarantees the eradication of germs and, consequently, the treatment of sickness.

 We use the theory of constrained optimization to achieve our objective of determining the optimal course of action. We allocate expenses to the use of antibiotics and to the active bacteria that survive. We want to lower the total of these expenses in order to apply antibiotics as little as possible while simultaneously eradicating the bacterial population.

 Let $u(t) = A_0 \delta(t)$, where $A_0$ gives the maximum concentration of antibiotic that can be applied and $\delta(t)$ is a function that controls the precise timing of the
 dosing protocol by determining the amount of antibiotic being let through at time $t$, $\delta(t) \in [0,1] \forall t \ge 0 $. By allocating expenses to each dosage procedure based on their antibiotic usage and killing efficiency, this can be stated mathematically.

 The target now is to determine $\delta(t)$ that minimizes the functional
\begin{equation}
J (B_s, B_p, B_r, L, \delta) = \int_0^t \int_0^{L(t)} (B_s+ B_p+ B_r) \, dx \, dt 
+ \int_0^t L\, dt
+ \int_0^t \frac {\beta}{2} \delta^2(t)  \, dt.
\label{eq:cost_func}
\end{equation}

subject to the PDE system above, where $\beta$ here is the suitable cost sensitivity parameter.   
Similar functionals have been utilized widely in the literature, including \cite{OCG1, OCG2, OCG3, OCG4}, with notable results.
 First, we prove the following theorem whose proof can be seen in Appendix.

\begin{theorem}
A unique control function $\delta^*(t)$ to minimize the cost functional can be determined as defined in\eqref{eq:cost_func}.

\end{theorem}

\begin{proof}
Let $\mathbf{r}(x,t,y,z)$
 be the $RHS$ of the system. We note the following axioms about our model and control function $\delta(t)$.
\begin{enumerate}
    \item $r(t,x,u)$ satisfies the Lipschitz property.
    \item Since at least the constant dose is admissible, hence the set of all acceptable control functions are non-empty.
    \item The control function's range $[0,1]$ is closed, compact and convex.
    \item The integrand used in equation~\eqref{eq:cost_func} is convex w.r.t the control function (the term $\beta \; \delta^2 (t)$).
    \item The dirichlet boundary condition $A(L,t) = u(t)= A_0 \; \delta(t)$ is linear in control function.
\end{enumerate}

  Since, our system admits the above five properties hence the corollary 4.1 of \cite{OC5} can be applied and for our situation, we determine that there is a single optimal control function $\delta(t)$ that minimizes the objective cost functional. 
\end{proof}
 Thus, the aforementioned theorem ensures that there is a special treatment plan that reduces\eqref{eq:cost_func}. To find this optimal function $\delta(t)$, we will use \textbf{Pontryagin's maximum principle}. Note that our system has parabolic pdes, as well as first order hyperbolic pdes. So, first, we give the basic theory from \cite{OC6}, \cite{OC7} that assures how PMP can be applied to our model, based on which we will procede for our model.

 \subsection{PMP for Mixed Parabolic–Hyperbolic Optimal Control Problem}

We discuss the PMP by considering  the system of two PDES, one parabolic and the other first order hyperbolic, where the control function appears as Dirichlet boundary condition on $y$. Later, we will extend this to our parabolic-heperbolic type model.  Consider the following system in the space-time domain 
\(
Q := (0,T) \times (0,L):
\)

\begin{equation}
\begin{cases}
y_t + A y + f_1(x,t,y,z) = 0, & \text{in } Q, \quad \text{(parabolic)}\\[2mm]
z_t + B z + f_2(x,t,y,z) = 0, & \text{in } Q, \quad \text{(hyperbolic)}
\end{cases}
\label{eq:general_system}
\end{equation}

with boundary and initial conditions:

\[
\begin{aligned}
y &= u \quad \text{on } \Sigma_y,\\
z &= 0 \quad \text{on } \Sigma_z,\\
y(0,x) &= y_1, \quad z(0,x) = y_2, \quad x \in (0,L),
\end{aligned}
\]

Here:

\begin{itemize}
\item $A$ is a second-order elliptic operator (parabolic operator), e.g.,
    \[
    A y = - \sum_{i,j=1}^N \frac{\partial}{\partial x_i} \left(a_{ij}(x) \frac{\partial y}{\partial x_j}\right),
    \]
    \item $B$ is a first-order hyperbolic operator, e.g.,
    \[
    B z = d(x) \frac{\partial z}{\partial x} + e(x) z,
    \]

    \item $u \in \mathcal{U}_{ad}$ is the control function and  $y_1, y_2$ are fixed numbers.
    \item The set of constraint on $u$ is defined by
    \[
    \mathcal{U}_{ad} = \left\{\, u \in L^{1}(S)\; \big|\; u(s,t) \in K_U(s,t) \subset U\ \text{for a.e. } (s,t) \in S \,\right\},
    \]
where $K_U(\cdot)$ is measurable multifunction with nonempty and closed values in $\mathcal{P}(\mathbb{R})$, and $U$ is a compact subset of $\mathbb{R}$.

\end{itemize}

The optimal control problem $(P)$ is to minimize the general cost functional:

\begin{equation}
J(y,z,u) = \int_Q F(x,t,y,z) \, dx \, dt 
+ \int_{\Sigma_y} G_y(s,t,u) \, ds \, dt 
+ \int_0^L L_0(x,y(T,x),z(T,x)) \, dx.
\end{equation}
subject to the PDE system above.

Define the Hamiltonian densities for the problem~\eqref{eq:general_system}:
\[
\begin{aligned}
\mathcal{H}_Q(x,t,y,z,p,q) &= F(x,t,y,z) + p\, f_1(x,t,y,z) + q\, f_2(x,t,y,z),\\
\mathcal{H}_{\Sigma_y}(s,t,u,p) &= G_y(s,t,u) + p\, u,\\
\mathcal{H}_0(x,y(T),z(T),p,q) &= L_0(x,y(T),z(T)) - p\, y_1 - q\, y_2.
\end{aligned}
\]
where $p$ and $q$ are adjoint variables for the system~\eqref{eq:general_system}.

\begin{theorem} Let $(\bar y, \bar z, \bar u)$ be an optimal solution. Then there exist adjoint variables $p,q$ that satisfy the pdes:

\begin{equation}
\begin{cases}
- p_t + A p + f_{1y} p + f_{2y} q = F_y(x,t,\bar y, \bar z), & p(T,x) = L_{0y}(x),\\[1mm]
- q_t + B q + f_{1z} p + f_{2z} q = F_z(x,t,\bar y, \bar z), & q(T,x) = L_{0z}(x),
\end{cases}
\label{eq:general_adjoint}
\end{equation}

and the optimal boundary control satisfies the minimum principle:

\begin{equation}
\mathcal{H}
_{\Sigma_y}(s,t,\bar u,\frac{\partial p (s, t)}{\partial n_A}) = \min_{u \in U_{ad}} \mathcal{H}
_{\Sigma_y}(s,t,u,\frac{\partial p (s, t)}{\partial n_A}), 
\quad \text{for almost all } (s,t) \in \Sigma_y.
\end{equation}
\end{theorem}
Define \[ \mathcal{H}
(x,t,y,z,u,p,q,s) =\mathcal{H}
_Q(x,t,y,z,p,q) + \mathcal{H}
_{\Sigma_y}(s,t,u,p) + \mathcal{H}
_0(x,y(T),z(T),p,q), \]
then the optimal control function $\delta^* (t)$ can be found using the optimality condition

\begin{equation}
\frac{\partial \mathcal{H}
}{\partial \delta (t)} =0 
\label{eq:general_opt_cond}
\end{equation}

Now we apply the above theory to our model.

The cost functional for our model is 
\begin{equation}
J (B_s, B_p, B_r, L, \delta) = \int_0^t \int_0^{L(t)} (B_s+ B_p+ B_r) \, dx \, dt 
+ \int_0^t L\, dt
+ \int_0^t \frac {\beta}{2} \delta^2(t)  \, dt.
\end{equation} we used $\delta (t)$ instead of $u(t)$  as $u(t) = A_0 \; \delta (t)$; $A_0$ is constant and $\delta(t) \in [0,1]$ is to be controlled for optimal results.

The optimality condition gives us
\begin{equation}
\delta (t) = Proj_{[0,1]} \left(-\frac{A_0}{\beta} \, \frac{\partial p_5(L,t)}{\partial x} \right) =
\min\!\left( \max\!\left( 0, -\frac{A_0}{\beta} \, \frac{\partial p_5(L,t)}{\partial x}\right), 1 \right).
\label{eq:our_opt_cond}
\end{equation}
Here we took projection on $[0,1]$ as $\delta(t) \in [0,1]$.

Now we solve the system of pdes consisting of our model~\eqref{eq:Bs}-~\eqref{eq:L}, the adjoint equations and the optimality condition ~\eqref{eq:our_opt_cond} to get the profile of the control function $\delta(t)$.

\subsection{Numerical Scheme}

The coupled PDE system is solved using a mixed implicit-explicit (IMEX) finite difference scheme. The diffusion-dominated equations for $A$ and $C$ are treated implicitly (Crank--Nicolson type) for stability, while the advection-dominated equations for $B_s, B_p, B_r,B_d$ are solved explicitly using a first-order upwind scheme. The adjoint equations are discretized similarly, with the parabolic adjoint implicit and the hyperbolic adjoint explicit, backward in time. The time step is chosen to satisfy the CFL condition for the hyperbolic component, \;  
\(
\Delta t \le \frac{\Delta x}{\max |v(x)|},
\)
ensuring stability of the advection scheme. Neumann and Dirichlet boundary conditions are applied as appropriate, providing a stable and accurate solution for the mixed parabolic--hyperbolic optimal control problem. We used $\Delta t= 0.001$ and $\Delta x= 0.005$. We selected the penalty parameter $\beta$ to be $0.001$.   The numerical scheme was implemented in Python.

This mixed explicit–implicit (IMEX) approach balances stability, accuracy, and efficiency, and is a standard method for mixed hyperbolic–parabolic PDE systems. Using fully implicit for both advection and diffusion would be more computationally expensive, and implicit advection can introduce unnecessary numerical damping, so the IMEX approach is more efficient and accurate for this system.

The numerical results obtained are shown in Figure~\ref{fig:myimage7} for continuous dosing and in Figure~\ref{fig:myimage8} for periodic dosing with a poeriod of 6 hours.

\subsection{Results of Treatment through Optimal Dosing}
\begin{figure}[H]
    \centering
    \includegraphics[width=0.7\textwidth]{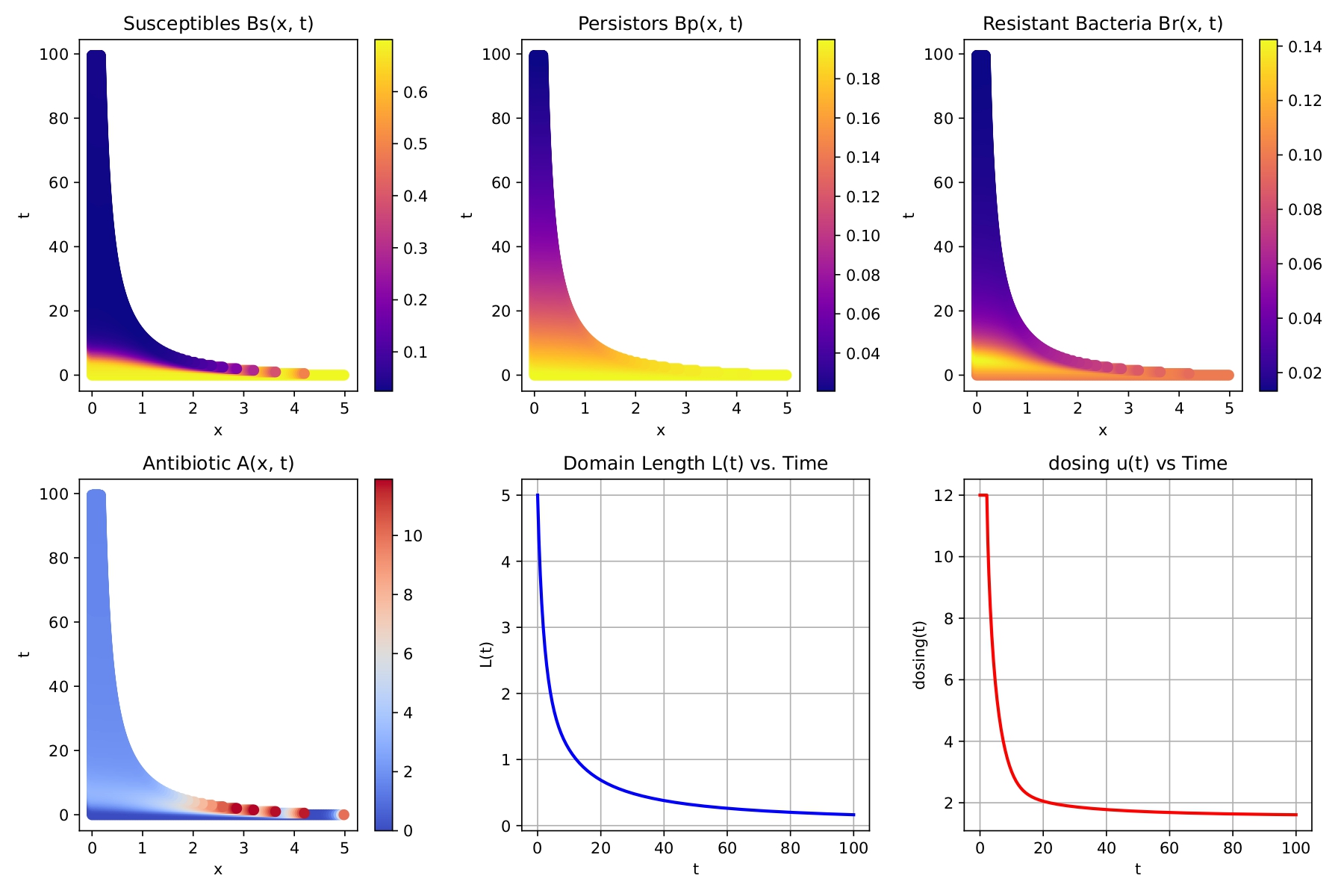} 
    \caption{Optimal Control for Continuous Dosing (Total antibiotic used: 250 units)}
    \label{fig:myimage7}
\end{figure}
\begin{figure}[H]
    \centering
    \includegraphics[width=0.7\textwidth]{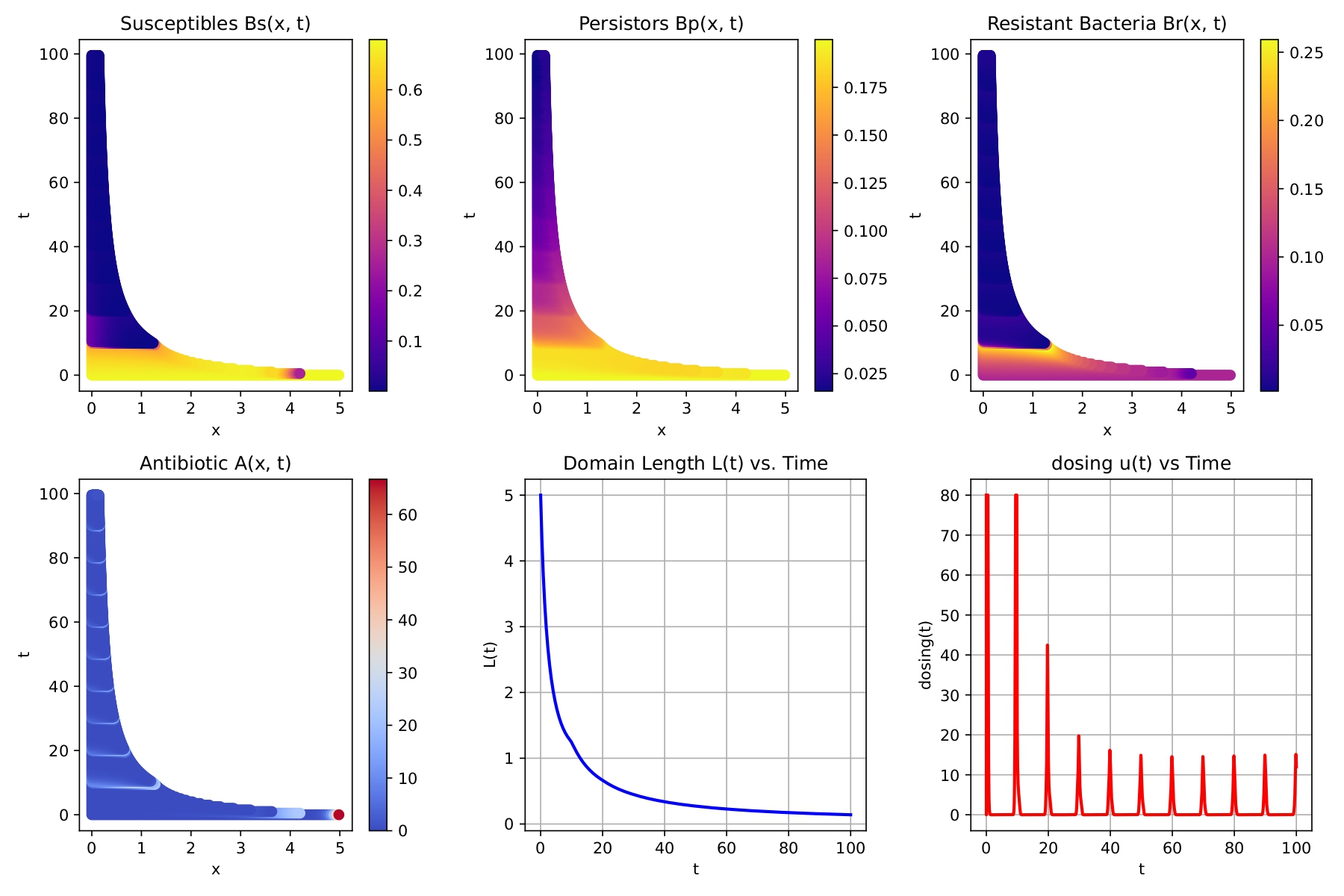} 
    \caption{Optimal Control for Periodic Dosing (Total antibiotic used: 210 units)}
    \label{fig:myimage8}
\end{figure}
We applied the Pontryagin’s Maximum Principle (PMP) to derive the optimal control function \(\delta(t)\) at the Dirichlet boundary. The optimized control profiles revealed a significant improvement in dosing efficiency. First we solve the system to get biofilm eradication with the continuous dosing by starting from an initial dose of 12 units. The control function \(\delta(t)\) gradually decreases over time and reduce the dosage over time as $u(t)=A_0 \delta(t)$. The total antibiotic usage under optimal control was recorded about 250 units, which is significantly lower as compared to 1000 units required previously without optimization.

Similarly, in the periodic dosing scenario, the optimal control yielded an effective eradication of the biofilm. 
The initial dose was applied to 80 units, while the total antibiotic used for the successful treatment was recorded to just 210 units. This is a significant improvement over the non-optimal periodic dosing, which required 640 units to get the same therapeutic result. Thus, using periodic dosing strategy under optimal control minimizes the total antibiotic usage, making it the most favorable and effective way for maintaining the complete bacterial clearance and disease treatment. These findings emphasize the need of implementing effective control measures, particularly those involving periodic dosing. The combination of temporal dose modulation and perfect control results in a physiologically realistic and resource-efficient biofilm treatment strategy.
\subsection{Dose Tapering}
We note that the optimal dosing strategy involves a dose tapering, such dosing techniques have been known to be effective in treatment with antibiotics \cite{OC1, ADNK}. Long-term use of high doses of antibiotics may have negative side effects, and quitting therapy too soon may cause the illness to recur. As a result, it has been proposed that antibiotic dosages be gradually decreased. However, one should exercise caution because the bacteria may begin to grow if the dosage is lowered below a specific threshold. We have demonstrated this using our model. We define $u(t) = A_0 \; \delta(t)$ where $\delta(t)$ determine the peak (amount of dose) over time. The disease may not be successfully eradicated by a given dosage, but it may reappear if the dosage is lowered below a particular threshold.
Therefore, the dosage should be suitably reduced to ensure that the bacteria is totally eradicated.

\subsection{Effects of \texorpdfstring{$k_d^r$}{kdr} and \texorpdfstring{$\beta$}{beta} on Optimal Control Strategy} 
We discuss the influence of varying the killing rate of resistant bacteria $k_d^r$ and the cost sensitivity parameter of antibiotic dosing $\beta$. We note from the Figure~\ref{fig:1Dimage5} that for low killing rate, we need more dose for the treatment and vise versa, which is consistent with trends noted in clinical practice.

\begin{figure}[H]
    \centering
    \includegraphics[width=0.7\textwidth]{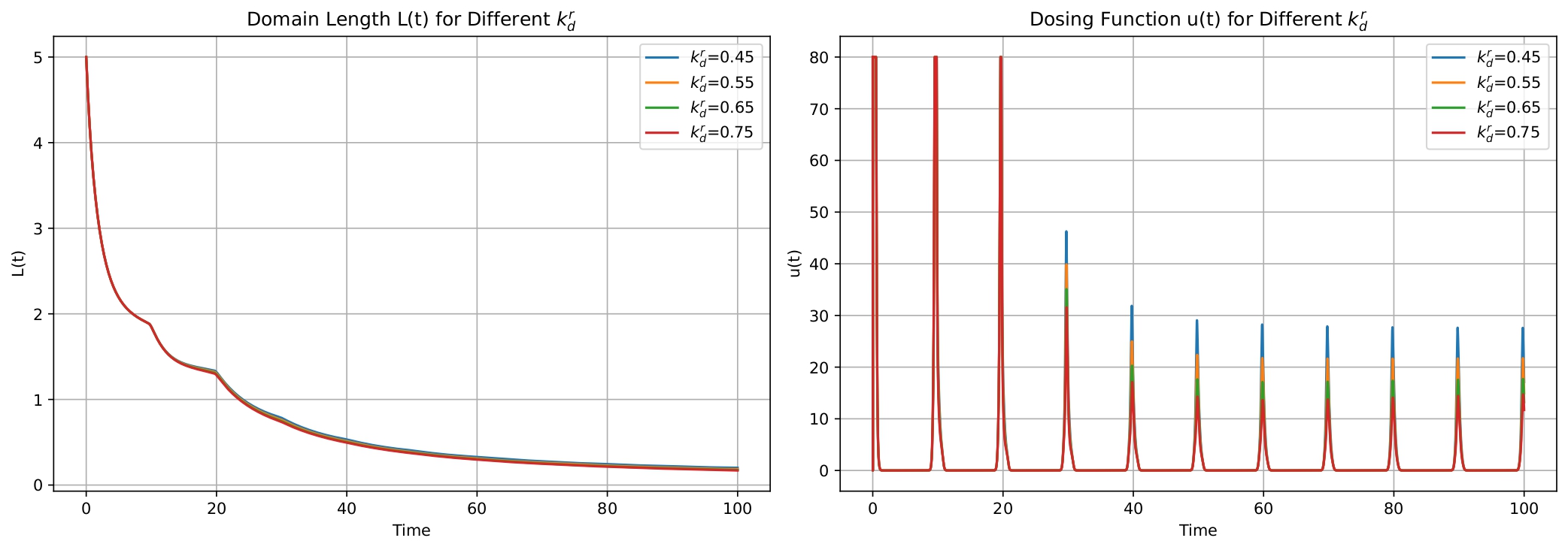} 
    \caption{Effect of varying $k_d^r$, at $k_d=1$}
    \label{fig:1Dimage5}
\end{figure}

\begin{figure}[H]
    \centering
    \includegraphics[width=0.7\textwidth]{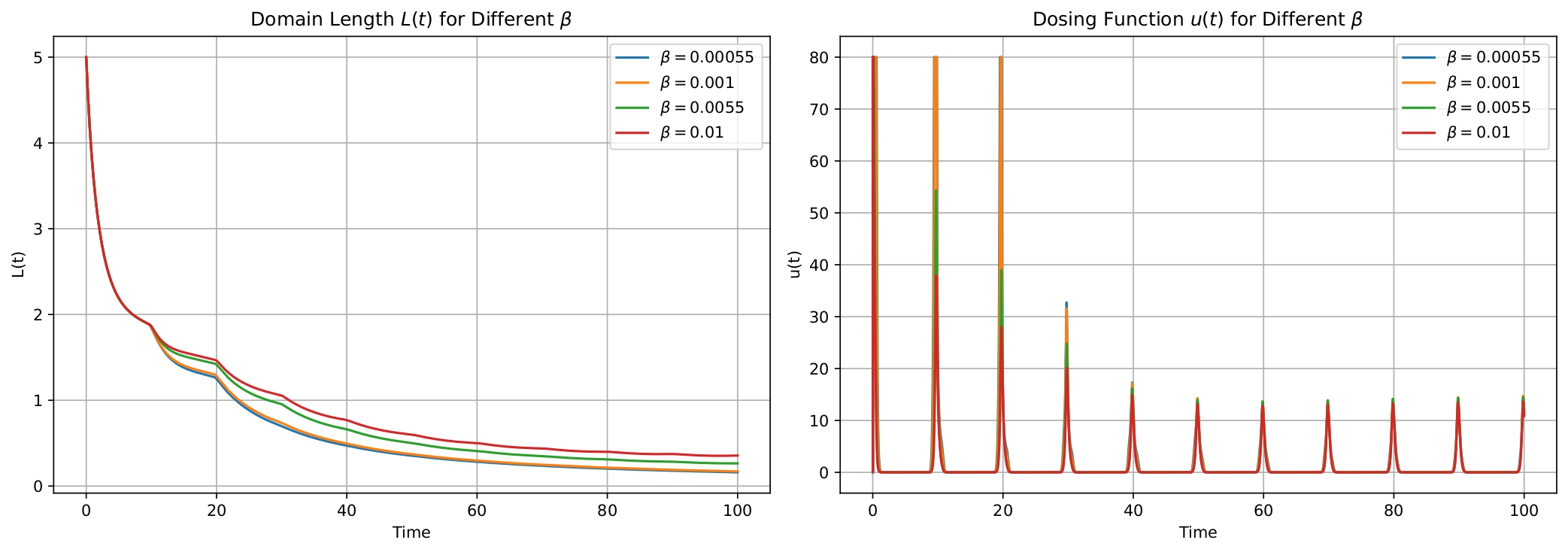} 
    \caption{Effect of varying cost sensitivity parameter $\beta$}
    \label{fig:1Dimage6}
\end{figure}
In~Figure~\ref{fig:1Dimage6}, we examine the effect of varying the cost sensitivity parameter on antibiotic dose and bacteria for various values of $\beta= 0.00055, \; 0.001, \; 0.0055, \; 0.01.$ 
They show a trade-off between antibiotic use and bacterial decrease. For lower values of $\beta$, the optimal control suggests larger antibiotic doses due to which there is a significant reduction in the bacterial population and a remarkable shrinkage of the biofilm, both bacteria and biofilm length approaching $0$.
But, as $\beta$ increases, the optimal control setup requires the reduction of the antibiotic applied. In all the cases, Pontryagin's principle permits the dose tapering  pattern, for which the dose gradually reduces over the dosing period. This behavior contrasts with the non-optimal treatment strategies discussed in section~\ref{sec:numerical_results}.

\section{Conclusion}
\label{sec:conclusions}
In this study, we modeled the transfer of resistance from resistant to susceptible bacteria via HGT within a 1D biofilm and studied different antibiotic dosing strategies. We proposed a mathematical model that included susceptible, persistor, and resistant bacteria dynamics along with the dynamics of the nutrient and antibiotics. The length of the biofilm was considered to depend on the bacterial population, and the bacteria may be eradicated by the use of antibiotics in which case the biofilm length goes to zero. \\
We determined the steady states of the model and their stability.
We found three biologically possible and feasible steady states of our system, one trivial $(0,0,0)$ and other 2 non-trivial states $(b_s^e, b_s^p, 0)$,  $(b_s^a, b_s^a, b_r^a)$. We derived analytic expressions for the sterile and co-existence states and verified our results numerically by using an appropriate finite difference scheme. In terms of model parameters, we get trivial states for $m_1<m_1^*$ and $m_2<q$, while non-trivial states when one of the conditions $m_1>m_1^*$ and $m_2>q$ holds. Biologically, the non-trivial steady states are obtained when the antibiotic is insufficient to kill the bacteria, while trivial state is obtained when antibiotic is increased beyond some specific value.\\
We observed the three steady states in the presence of antibiotics, both for continuous and periodic dosing. For continuous dosing, we started with a dose of 1 unit, and got treatment failure with steady state $(b_s^a, b_s^a, b_r^a)$. When we increased the dose to $2.5$ units, we got other non trivial steady state $(b_s^e, b_s^e, 0)$. Finally when we increased the dose to $10$ units, we got biofilm eradication and hence the trivial steady state $(0,0,0)$. 
Similarly we got non trivial steady states $(b_s^a, b_s^a, b_r^a)$ and $(b_s^e, b_s^e, 0)$ for periodic dosing of 5 and 8 units respectively with a dosing period of 6 hours, and the trivial steady state occurs by increasing the periodic dose to 40 units with the same period. 
As expected periodic dosing is better than continuous dosing, a possible reason being that it allows persistor bacteria to transform to susceptible form during off-dosing intervals, making them more vulnerable to antibiotics. The total amount of antibiotic to treat the biofilm is also reduced as susceptible bacteria are eradicated more efficiently by antibiotics as compared to persistor and resistant bacteria. \\
The antibiotic dosage has to be determined very carefully, less antibiotic may lead to the treatment failure or temporary success for some time followed by bacterial growth, and high doses beyond what is needed may create other health problems as antibiotics may have side effects. Theorem~\ref{th:theorem3} gives the ranges according to which if dose is in the range given by~\eqref{eq:insuf_dose}, we will get non trivial steady state and thus treatment failure. While for getting trivial steady state and treatment success we should increase the dose to a minimum threshold given in~\eqref{eq:suf_dose}.\\  
As mentioned determining the optimal antibiotic dosing is very important, and for this we employed an extended version of Pontryagin’s maximum principle (PMP), that minimizes the antibiotic dose, and also ensures the eradication of the biofilm. 
We employed PMP for both continuous as well as periodic dosing regimens. 
We defined our control function as $u(t)= A_0 \delta(t)$ and calculate $\delta(t)$ using PMP, where $\delta(t) \in [0,1]$ controls the antibiotic required at any instant of time. We conclude that a tapered dosing is optimal, ensures the eradication of biofilm while keeping the total antibiotic quantity low. We determine this for a variety of scenarios and observe that the qualitative behavior of the optimal dosing remains the same across parameter values.\\
As an extension of the current work we would like to consider a 2-D biofilm model, and study resistance transfer and efficient dosing strategies. Another direction would be to consider the fluid compartment as well, relaxing the saturation assumption we have used here. Finally, we would also like to explore HGT using transformation and transduction in biofilms and study antibiotic dosing in those cases.

\end{document}